\documentclass[journal]{IEEEtran}
\usepackage{citesort}
\usepackage[dvips]{graphicx}
\usepackage[cmex10]{amsmath}
\usepackage{algorithm}
\usepackage{algpseudocode}
\usepackage{array}
\usepackage[tight,footnotesize]{subfigure}
\usepackage{setspace}
\usepackage{enumerate} 
\usepackage{url}
\usepackage{empheq}
\usepackage{amsfonts, amssymb, amsthm, mathrsfs}
\usepackage{setspace}
\usepackage{Tabbing}
\usepackage{fancyhdr}
\usepackage{lastpage}
\usepackage{extramarks}
\usepackage{chngpage}
\usepackage{color}
\usepackage{indentfirst}
\usepackage{float,wrapfig}
\usepackage{times}
\usepackage{multirow}
\usepackage{bm}
\usepackage{bbm}
\usepackage{mathtools}
\usepackage{slashbox}
\usepackage{lipsum}

\DeclareMathOperator*{\argmax}{arg\,max}

\newtheorem{defi}{Definition}

\newtheorem{prop}{Proposition}

\newcommand{\sinr}{\textrm{SINR}} 
\newcommand{\mb}{\mathbf}

 \newcommand{\Csize}{L}

\long\def\comment#1{}

\newfont{\bbb}{msbm10 scaled 700}

\newfont{\bb}{msbm10 scaled 1100}

\newcommand{\Cc}{{\cal C}}

\newcommand{\Jc}{{\cal J}}

\newcommand{\Sc}{{\cal S}}

\newcommand{\Uc}{{\cal U}}

\begin{document}

\title{Harmonized Cellular and Distributed Massive MIMO: Load Balancing and Scheduling}

\author{\IEEEauthorblockN{Qiaoyang Ye$^*$, Ozgun Y. Bursalioglu$^\dagger$, Haralabos C. Papadopoulos$^\dagger$}

\IEEEauthorblockA{$\ ^*$Dept. of ECE, The University of Texas at Austin\\
$\ ^\dagger$Docomo Innovations Inc\\
Email: qye@utexas.edu, \{obursalioglu, hpapadopoulos\}@docomoinnovations.com.}}

\maketitle
\begin{abstract}
Multi-tier networks with large-array base stations (BSs) that are able to operate in the ``massive MIMO'' regime are envisioned to play a key role in meeting the exploding wireless traffic demands.
Operated  over small cells with reciprocity-based training, massive MIMO promises large spectral efficiencies per unit area with low overheads. Also, near-optimal user-BS association and resource allocation are possible in cellular massive MIMO HetNets using simple admission control mechanisms  and rudimentary BS schedulers, since scheduled user rates can be predicted \emph{a priori} with massive MIMO.

Reciprocity-based training naturally enables coordinated multi-point transmission (CoMP), as each uplink pilot inherently trains antenna arrays at all nearby BSs. In this paper we consider a  distributed-MIMO form of CoMP, which improves cell-edge performance without requiring channel state information exchanges among cooperating BSs. We present methods for harmonized operation of distributed and cellular massive MIMO in the downlink that optimize resource allocation at a coarser time scale across the network.  We also present scheduling policies at the resource block level which target approaching the optimal allocations. Simulations reveal that the proposed methods can significantly outperform the network-optimized cellular-only massive MIMO operation (i.e., operation without CoMP), especially at the cell edge.
\end{abstract}


\section{Introduction}\label{sec:introduction}
The exponential growth in wireless traffic is driving the densification of cellular networks.  Existing networks of carefully planned conventional macro base stations (BSs) are becoming transformed into dense irregular heterogeneous networks (HetNets), as they are  continuously supplemented with various types of BSs, differing in transmit power, physical size, and deployment cost~\cite{And7ways13}.

It has been well recognized that traditional user-BS association schemes are highly suboptimal for HetNets, due to the large disparities in BS transmit power~\cite{AndSin14}. Moreover, the non-uniform user distribution and irregular deployment of small BSs make load balancing critical.  
Various approaches have been used to investigate load balancing in HetNets \cite{AndSin14}, including stochastic geometry approach \cite{SinAnd12} and techniques from game theory \cite{AryKes13}. Some standardization efforts have also been made for load balancing in HetNets, e.g., in the form of \emph{cell range expansion} \cite{DamMon11}.

Several recent works \cite{YeAndLB13,YeAndABS13,BedAgr13,BetBur14a} recast  load balancing in HetNets as a network utility maximization (NUM) problem.
Paper \cite{YeAndLB13} studies the optimal user-BS association problem  in HetNets and shows  a great improvement in user rate distribution  via systematic load balancing.  Papers \cite{YeAndABS13,BedAgr13} consider the joint optimization of user association and BS muting -- referred to in 3GPP as enhanced intercell interference coordination (eICIC).

In parallel,  there is surging interest in equipping BSs with large antenna arrays.  With higher-frequency spectrum becoming available,  large arrays become feasible even for small cells, as more effective antennas can be packed into a small form factor\footnote{For example, at 5GHz,  $49$ antennas (arranged on square grid at half wavelength spacing) can be packed on a 20cm $\times$ 20cm antenna patch.}.  
By exploiting channel reciprocity, massive arrays can be trained in the uplink   (whether for uplink or downlink transmission) with low overheads \cite{Mar10}.
This enables very large spectral efficiencies per unit area  via massive MIMO, i.e., via serving simultaneously many users (although much fewer than antennas), each at a very high rate \cite{Mar10,HuhCai11,HoyTen13,LarEdf14}.  Attributes of massive MIMO can also be exploited to achieve near-optimal load balancing over massive MIMO HetNets using simple user-BS association methods with cellular transmission (where data for each user is transmitted from a single BS)  \cite{BetBur14a}.


In this paper we consider the use of coordinated multi-point transmission (CoMP) as a means for improving network performance, in particular, the cell-edge performance. CoMP is naturally enabled by reciprocity-based training, since  a single uplink pilot from a user terminal can train all nearby antennas. In regular cellular layouts with massive MIMO BSs, \cite{HuhCai11} shows gains to cell-edge users via CoMP.

In this paper we focus on a  distributed-MIMO form of CoMP, which does not require channel state information (CSI) exchanges among cooperating BSs and  allows us to  develop a systematic approach for allocating resources for CoMP and cellular  transmission. 
The methods we develop are based on formulating a NUM problem with respect to user association and resource allocation via extensions of the  framework for cellular transmission developed in \cite{BetBur14a}. We also present scheduling policies at the resource block (RB) level, which target approaching these optimized (coarser time scale) resource allocations. Simulations show that the proposed  harmonized CoMP/cellular operation can  provide significant gains with respect to cellular-only massive MIMO operation \cite{BetBur14a}, especially at the cell edge.

\section{System Model}\label{sec:model}
We consider the downlink of a cellular network comprised of $J$  BSs and $K$ single-antenna users.  We use $j\in \Jc=\{ 1,2, \ldots, J\}$ and $k\in \Uc=\{ 1,2, \ldots, K\}$ to index BSs and users, respectively.  We let $M_j$ denote the number of antennas at BS $j$ and assume $M_j\gg 1$. We assume time division duplex (TDD) operation with reciprocity-based CSI acquisition \cite{Mar10,HuhCai11}.  Hence, every BS antenna in the vicinity of user $k$ can estimate its downlink channel coefficient to user $k$ from the uplink pilot transmitted by user $k$.  This enables the training of large antenna arrays (e.g., $M_j \gg 1$) with pilot overheads 
proportional to number of simultaneously served users \cite{Mar10}. In contrast to feedback-based CSI acquisition, it also allows a user terminal to train {\em multiple nearby} BSs, and enables CoMP  without  additional training overheads.

Transmission resources are split into  slots or RBs, with each RB corresponding to a contiguous block of OFDM subcarriers and symbols. 
In any  given RB, we let $P_j$ denote the transmit power at BS $j$, and assume that this power is  equally split among all served user streams.
We assume a block-fading channel model where the user-BS channels remain constant within each RB \cite{caire2010multiuser,Mar10,HuhCai11,HoyTen13}.
We let  $\mathbf{g}_{kj}=\sqrt{\beta_{kj}}\, \mathbf{h}_{kj}$  denote the $M_j\times 1$ channel vector between BS $j$ and user $k$ on a generic RB, with the slow-fading scalar $\beta_{kj}$ characterizing the combined effect of distance-based path loss and location-based shadowing, and the vector  $\mathbf{h}_{kj}$ capturing fast fading. We assume that the vectors $\mathbf{h}_{kj}$'s are independent in $k$ and $j$, and that $\mathbf{h}_{kj}$ has  i.i.d.  $\mathcal{CN}(0,1)$  elements (independent Rayleigh fading).  We also assume that the thermal noise process at user $k$ is i.i.d. with $\mathcal{CN}(0,\sigma^2)$ samples.

\section{MIMO Transmission}
Within each RB, a subset of users are {\em active}, i.e., are scheduled for transmission.
The coded data for any given scheduled user can  be transmitted either from a single BS via \emph{cellular transmission}, or from multiple BSs via a CoMP mode referred to as  \emph{distributed MIMO transmission}. 

\subsection{Prior Art: Cellular Massive MIMO \cite{BetBur14a}}
In setting the stage for the distributed MIMO operation presented in this work, it is worth revisiting load balancing and scheduling for cellular massive MIMO, as considered in \cite{BetBur14a}. Let $S_j$ denote the maximum number of users served by BS $j$ on any given RB, with $S_j \ll M_j$. 
Under mild assumptions on fading,  the achievable user instantaneous rates on RB $t$, $r_{kj}(t)$,  can be predicted \emph{a priori} in the massive MIMO regime \cite{BetBur14a}. In particular, there exist deterministic quantities $\{r_{kj}\}$ such that  $r_{kj}(t) \stackrel{{\rm a.s.}}{\rightarrow} r_{kj}$, for all $k \in \Uc$ and $j \in \Jc$,  as $M_j, S_j \rightarrow \infty$, with fixed  $ \nu_j = S_j/M_j \geq 0$  \cite{Mar10,HoyTen13,HuhCai11}. This convergence  is very fast with respect to the $M_j$'s.

Letting  $\Sc_j(t)$ denote  the set of users served by BS $j$ on RB $t$ and $x_{kj} = \lim_{T \rightarrow \infty} \frac{| \{ t : k \in \Sc_j(t)\}| }{T}$  denote  the \emph{activity fraction} of user $k$  on BS $j$, the long-term averaged  throughput of user $k$ can be obtained via \cite{BetBur14a} 
\begin{equation} \label{throughput1}
R_k = \sum_{j \in \Jc} x_{kj} r_{kj},  \  \forall  k \in \Uc.
\end{equation} 

The advantages of the approach in \cite{BetBur14a} for cellular massive MIMO operation can be summarized as follows:
\begin{enumerate}[(A)]
\item The $r_{kj}$'s  are accurate  peak-rate proxies, which are independent of scheduled instances and user sets. 
\item User throughputs depend on activity fractions, via (\ref{throughput1}).  
\item The (combinatorial) user-cell association problem is recast as a (convex) NUM problem with respect to the $\{x_{kj}\}$ variables, subject to  resource constraints.
\item Any  $\{x_{kj}\}$ set not violating any resource constraints can be realized by a suitably designed scheduler.
\end{enumerate}

\subsection{CoMP via Distributed MIMO} 
The distributed MIMO scheme we consider corresponds to a form of CoMP that allows harvesting performance gains at the cell edge, with low operational overheads.
\begin{defi}\label{admissible-dmimo}{\bf Admissible Distributed MIMO Schemes:}
An admissible distributed MIMO scheme is a scheme that schedules transmissions for users on a sequence of RBs and, on each RB, satisfies the following:
\begin{enumerate}[(i)]
\item All the users served by a given BS $j$ are served in clusters of the {\em same} size $\Csize$ for some $\Csize\ge 1$.  
\item BS $j$ serves  at most $S_j(\Csize)$ users, for some fixed $S_j(L)$, satisfying $ S_j \le S_j(\Csize)\le \Csize S_j$. 
\item The user beams  (i.e., the precoding vectors) at BS $j$ are designed as if BS $j$ were engaging in cellular  MU-MIMO transmission over all the users it serves. 
\item All BSs serving a user transmit the same coded user stream. Each BS transmits the stream on a beam that is (independently) designed for the users at that BS.
\item $M_j \gg S_j(\Csize)$, for all $\Csize$ and $j$ considered.
\end{enumerate}
\end{defi}
We also assume that, within each RB, the transmit power at each BS is equally split among scheduled users.  


Table~\ref{sample-RBs} provides an example of a scheme complying with Defn.~\ref{admissible-dmimo},  assuming BS clusters of size 1 (cellular transmission) and 2. Four BSs are considered with $P_j=1$, $S_j(1)=S_j=2$, and $S_j(2)=3$.  As the table reveals, each BS on RB \#1 engages in cellular transmission. 
On RB \#2, BSs pairs jointly serve triplets of users. RBs \#3 and \#4 provide additional, more interesting,  modes.  No two users are served by the same BS cluster on RB \#3, while on RB \#4 BSs 1-2 jointly serve a triplet of users, while  BSs 3-4 serve users in cellular transmission. Note also that (at least) 8, 6, 6 and 7  uplink pilot dimensions are needed to enable RBs \#1, \#2, \#3 and \#4, respectively. Evidently,  the choice of scheduled user sizes, $S_j(\Csize)$,  signifies how aggressively pilot dimensions are reused across the network. 

\begin{table}
\caption{Example of RBs enabled by distributed MIMO over 4 BSs. } 
\vspace{-.2in}
\begin{center}
\begin{tabular}{|c||c|c|c|c|c|}
\hline
RB  & & BS 1 & BS 2& BS 3 & BS 4\\
\hline 
\hline
 & Cluster Size & 1 & 1 & 1 & 1\\
\hline
\#1 & User Power & 1/2 & 1/2 &  1/2 & 1/2\\
\hline 
& Served Users & 1,2 & 3,4 & 5,6 & 7,8 \\
\hline
\hline
 & Cluster Size & 2 & 2 & 2 & 2\\
\hline
\#2 & User Power & 1/3 & 1/3 &  1/3 & 1/3\\
\hline 
& Served Users & 1,2,3& 1,2,3 & 4,5,6& 4,5,6\\
\hline
\hline
 & Cluster Size & 2 & 2 & 2 & 2\\
\hline
\#3 & User Power & 1/3 & 1/3 &  1/3 & 1/3\\
\hline 
& Served Users & 1,2,3 & 1,4,5 & 2,4,6 & 3,5,6 \\
\hline
\hline
 & Cluster Size & 2 & 2 & 1 & 1\\
\hline
\#4 & User Power & 1/3 & 1/3 &  1/2 & 1/2\\
\hline 
& Served Users & 1,2,3 & 1,2,3 & 4,5 & 6,7\\
\hline

\end{tabular}
\end{center}
\label{sample-RBs}
\vspace{-.2in}
\end{table}

It is worth making a few remarks regarding the choice  of the distributed MIMO schemes of Defn.~\ref{admissible-dmimo}. First, the  schemes of Defn.~\ref{admissible-dmimo} provide the following CoMP benefits:
\begin{enumerate}[(i)]
\item  \textbf{Performance gains at the cell edge:} 
The beamforming (BF) gain provided by the cellular scheme (that  the distributed scheme is based upon) becomes intra-cluster BF gain in the distributed MIMO case, as the same coded data is transmitted from all BSs serving the user. Similarly, the intra-cell interference mitigation capabilities of the cellular scheme are extended across the cluster of BSs from which the user is served.  As a result, performance gains can be realized at the cell edge.
\item \textbf{Low training overheads:}  An uplink pilot  from a user terminal trains all nearby BS antennas, whether these are in one or many locations. Thus, CSI acquisition  between a user and nearby BSs need not  incur additional overheads with respect to cellular transmission.
\end{enumerate}

\noindent In addition, the schemes of Defn.~\ref{admissible-dmimo}  possess several important properties that are not in general present in  CoMP schemes:
\begin{enumerate}[(a)]
\item \textbf{Local precoding at each BS:} This is due to item (iii) in  Defn.~\ref{admissible-dmimo}.  For instance, in the case of linear zero-forcing beamforming (LZFBF), the beam for each user served by BS $j$  is chosen within the null space of the channels of all the other users served by BS $j$, no matter whether there are additional BSs serving the user on the same RB or not. 
\item \textbf{No need for CSI exchanges among BSs:}  Again, due to item (iii) in  Defn.~\ref{admissible-dmimo},  BS $j$ only needs CSI between the users it serves and the antennas of BS $j$ in order to generate the user beams at BS $j$.
\item \textbf{Flexible scheduling:}  The schemes of   Defn.~\ref{admissible-dmimo} enable user-specific BS-cluster transmission, and allow  serving users from overlapping but different clusters of BSs  on the same RB (see, e.g.,  RB \#3 in Table~\ref{sample-RBs}). 
\item \textbf{Simple predictors of instantaneous rates:} As shown in  \cite{HuhTul12}, the instantaneous user rates can also be predicted \emph{a priori} with CoMP. However, unlike general CoMP settings, where a user's instantaneous rate depends on the other users co-scheduled for transmission on the same RB  \cite{HuhTul12}, the schemes of Defn.~\ref{admissible-dmimo} make a user's instantaneous rate \emph{independent}  of the identities of the other users in the scheduling set. 
\end{enumerate}
 
As a result, the cellular-transmission attributes (A)--(C) exploited in  \cite{BetBur14a}  can be appropriately extended to allow resource allocation for  the schemes of Defn.~\ref{admissible-dmimo}, in the form of network-optimized activity fractions between users and clusters of BSs.  Although, as it turns out,  item (D)  is not always true with distributed MIMO, i.e., these activity fractions may not necessarily be realizable, as shown in Sec.~\ref{sec:scheduling-policies}, scheduling policies can be designed that may  approximate these fractions  in practice sufficiently well.

\section{Peak Rates and Scheduled Throughputs}\label{sec:analyze}
In this section, we develop proxy expressions for the instantaneous user rates and for the scheduled user throughputs that are provided by any  given scheduling policies enabling distributed MIMO transmission with either LZFBF or maximum ratio transmission (MRT). 

We consider a scheduling policy on RBs  $\{1, \, 2\, \cdots, T\}$  and assume that all the large-scale coefficients stay fixed within this period.  Any such scheduling policy can be described in terms of  the scheduling sets $\{ \Sc_\Cc(t); \ \  \forall \Cc, \  \forall t \in \{1, \, 2\, \cdots, T\} \}$, where $\Sc_{\Cc}(t)$ denotes the set of active users served by cluster $\Cc$ on RB $t$. Thus, the received signal at an active user $k\in \Sc_{\Cc}(t)$ on RB $t$ can be expressed by 
\begin{equation}\label{eq:receivesig}
\begin{aligned}
y_k(t) = &\underbrace{\sum_{j\in \Cc} \sqrt{\frac{P_{j}}{N_{j}}} \mb{g}_{kj}^H \mb{f}_{kj}s_k }_\text{desired} + \underbrace{ \sum_{j\in \Cc} \sum_{\substack{u\in \Sc_{\Cc}(t)\\ u\neq k}} \sqrt{\frac{P_{j}}{N_j}} \mb{g}_{kj}^H \mb{f}_{uj}s_u }_\text{intra-cluster interference} \\
+&\underbrace{\sum_{l\notin \Cc} \sum_{u\in \cup_{(\mathcal{C'}: l\in \mathcal{C'})} \Sc_\mathcal{C'}(t)} \sqrt{\frac{P_l}{N_l}} \mb{g}_{kl}^H \mb{f}_{ul}s_u }_\text{inter-cluster interference} + \underbrace{w_k}_\text{noise},
\end{aligned}
\end{equation}
where $\Cc$ denotes the cluster (set) of BSs serving user $k$ on RB $t$, $\mathcal{C'}$ denotes the cluster including BS $l$, $s_u$ denotes the unit-power stream for user $k$, and $\mathbf{f}_{uj}$ denotes the unit-norm precoding vector  for user $u$ at BS $j$.

Let $r_{k\Cc}$ denote the peak rate of user $k$ from BS cluster $\Cc$. It can be shown using the techniques in \cite{LimCha13,HuhTul12} that, with distributed MIMO based on LZFBF, $r_{k\Cc}$ is given by\footnote{Expression~(\ref{eq:rate-zf}) assumes that  $\forall j\in \Cc$, BS $j$ serves $S_j(\Csize)$ users, each user at power $P_j/S_j(\Csize)$. In the case that fewer users are served by one of the BSs, the LHS in (\ref{eq:rate-zf}) represents an achievable (lower bound) rate.}  
\begin{equation}\label{eq:rate-zf}
r_{k\Cc} \!= \!\log_2\!\left(\!\!1\!+\!\frac{\displaystyle \sum_{j\in \Cc} \sum_{\ell\in \Cc} \!\!\sqrt{\!P_{j}P_{\ell}\beta_{kj}\beta_{k\ell}b_{j}(|\Cc|) b_{\ell}(|\Cc|)}}{\sigma^2 + \sum_{ \ell\notin \Cc} P_{\ell} \beta_{k\ell}}\!\right)\!\!,
\end{equation}
where $b_{j}(L)=\frac{M_{j} - S_{j}(L)+1}{S_{j}(L)}$.
Similarly, for the case that the distributed MIMO transmission is based on MRT, 
\begin{equation}\label{eq:rate-mrt}
r_{k\Cc} \!=\! \log_2\!\left(\!1 + \frac{\sum_{j\in \Cc}\sum_{\ell \in \Cc}\sqrt{\frac{P_j P_\ell M_j M_\ell\beta_{kj}\beta_{k\ell}}{S_{j}(|\Cc|)\, S_{\ell}(|\Cc|)}} }{\sigma^2 + I_{k\Cc}+ \sum_{\ell\notin \Cc} P_\ell \beta_{k\ell}} \!\right)\!,
\end{equation}
where $I_{k\Cc}\!=\!\sum_{ \!j\in \Cc}\! \frac{S_{j}(\!|\Cc\!|\!)\!-\!1}{S_{j}(|\Cc|)} \!P_j\!\beta_{kj} $ is intra-cluster interference. For completeness, we provide the proof of (\ref{eq:rate-zf}) and (\ref{eq:rate-mrt}) in Appendixes \ref{pf:theo-zf} and \ref{pf:theo-mrt}, respectively.

Similar to cellular massive MIMO in \cite{BetBur14a}, the long-term user throughput  with the admissible distributed MIMO  schemes of Defn.~\ref{admissible-dmimo} can be expressed in terms of the distributed MIMO peak rates and the activity fractions provided by the scheduling policy.
In the limit $T\to \infty$, the throughout of user $k$  can be expressed as\footnote{Convergence to the limiting expressions of interest is very quick \cite{BetBur14a}.}
\begin{equation}
\label{avgtp-vs-actfracs}
R_k = \sum_\Cc x_{k\Cc}\, r_{k\Cc},
\end{equation}
where $x_{k \Cc} = \lim_{T\to \infty} \frac{|\{t: 1\le t\le T; \ k \in \Sc_\Cc(t)\}| }{T}$ is the  activity fraction of user $k$ with respect to cluster $\Cc$.

\begin{figure*}[tp]
\centering
\setcounter{subfigure}{0}
\subfigure[Illustration of user location.]{\label{fig:location-toyeg}\includegraphics[width=6.2cm, height=5.8cm]{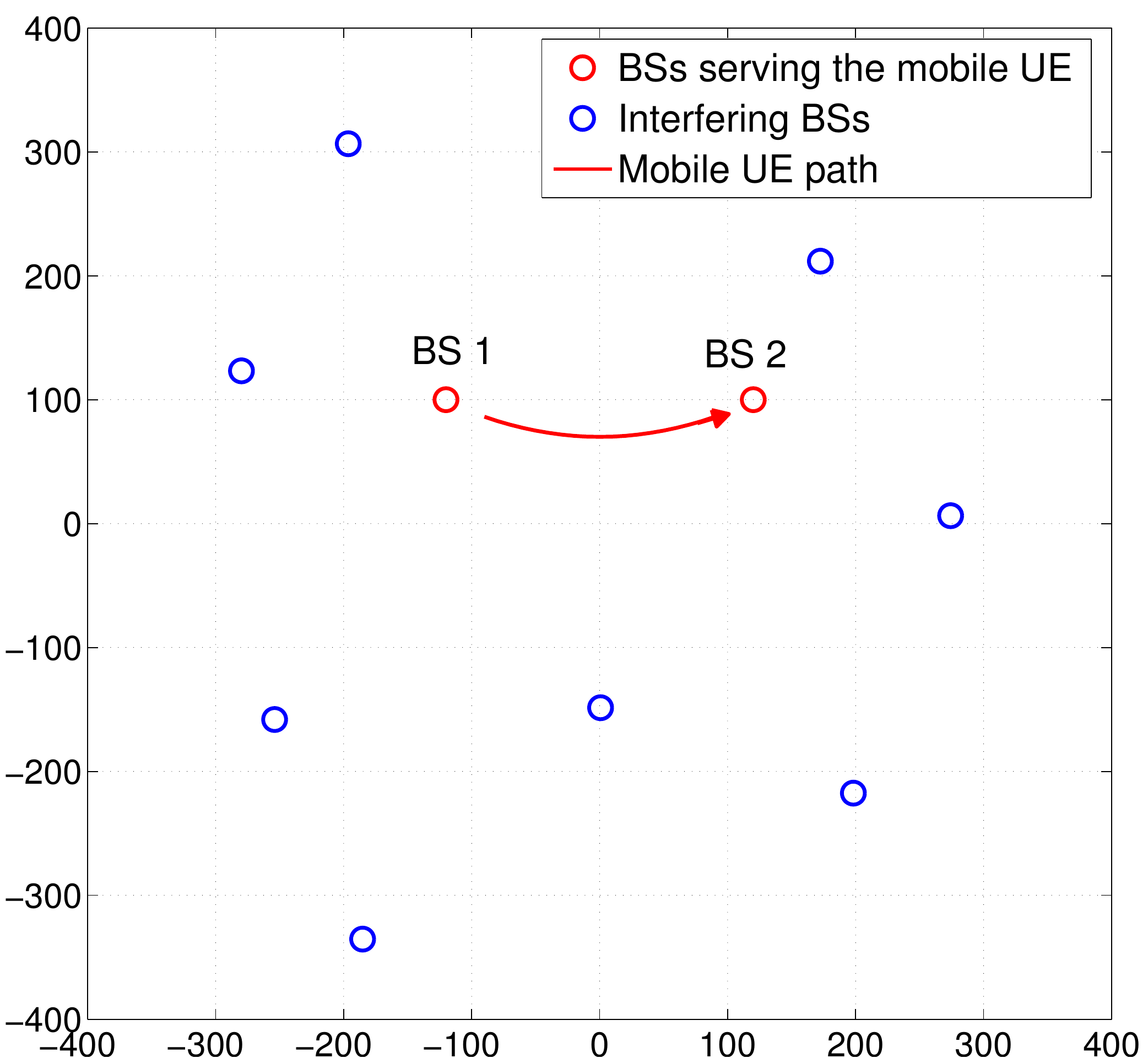}}
 \hspace{0.4in}
\subfigure[Instantaneous rate vs. x-axis coordinate in (a).]{\label{fig:rate-toyeg}\includegraphics[width=6.5cm, height=5.8cm]{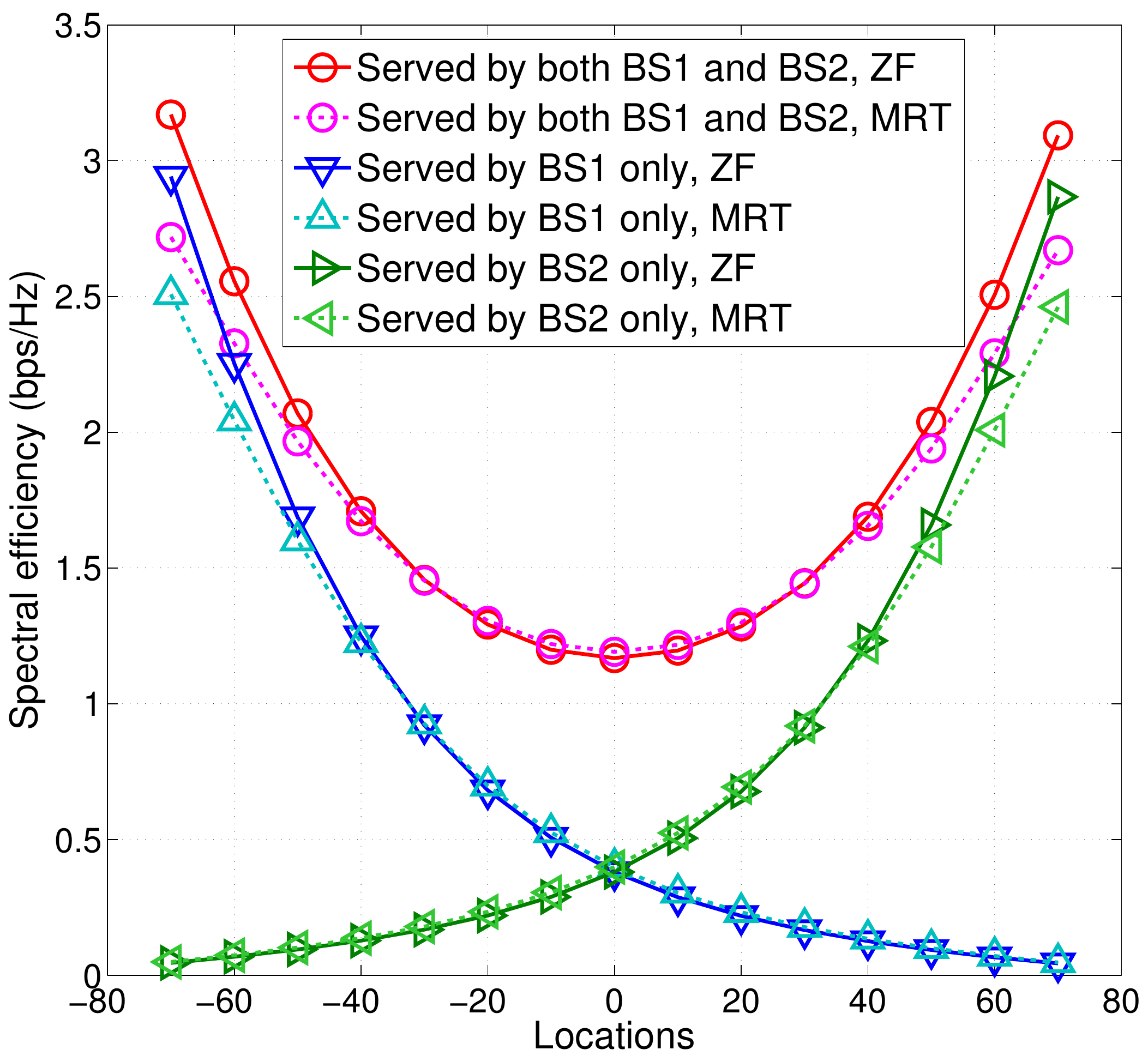}}
\setlength{\abovecaptionskip}{0pt}
\caption{Illustration of spectral efficiency versus user locations. The location in Fig. \ref{fig:location-toyeg} indicates the x-axis coordinate of the path in Fig. \ref{fig:location-toyeg}.  The cell edge users (e.g., the areas near the origin) benefit from distributed MIMO.}
\label{fig:toyeg}
\end{figure*}
%
Fig.~\ref{fig:toyeg} shows an example of the potential benefits that distributed MIMO transmission can offer, by showing  the peak rates of cellular vs pair-BS distributed MIMO transmission as a function of user location. 
When the user is close the transmitting BS cellular transmission is as good as anything. On the other bend when the user is close to the cell edge between BSs 1 and~2,   distributed MIMO transmission from the cluster  $\{1,2\}$ yields about 3 times higher rates than cellular transmission.  In the next section we formally consider the problem of allocating resources to users across BSs or clusters of BSs so as to optimize the network-wide system performance.

\section{User-Cluster Association as NUM}
In this section, we formulate the user-cluster association problem as a NUM  of activity fractions across all users, to optimize a network-wide utility function capturing the operator's notion of (inherently subjective) fairness.

Before formulating the NUM problem, it is worth restricting the domain of scheduling options in order to obtain solutions that are of practical interest. We focus on cluster sizes $\Csize \in \{1, \, 2, \cdots, \Csize_{\max} \}$ for some appropriately chosen maximum\footnote{The choice of $\Csize_{\max}$ is a design choice. It depends on the average number of nearby BS arrays that users typically see and the complexity that can be afforded. In our simulations, we set $\Csize_{\max}=4$.}  cluster size,  $\Csize_{\max}$. Motivated by the example in Table~\ref{sample-RBs}, we consider the following architecture. 

\begin{defi}\label{first-arch}
{\bf Uniform Cluster-Size Architecture (UCS)}: 
A scheme from Defn.~\ref{admissible-dmimo} is a UCS architecture, if for each
 $\Csize\in \{1,\,2,\,\ldots,\, \Csize_{\rm max}\}$, a $\lambda_\Csize\ge 0 $ fraction of the RBs is allocated to serving size-$L$ clusters, and if on any RB from this $\lambda_\Csize$ fraction the following are satisfied:

\noindent (i) each scheduled user is served by a (user-dependent) cluster of $\Csize$ BSs;

\noindent (ii) for each $j\in \Jc$, BS $j$  serves no more than $S_j(\Csize)$ users.
\end{defi}
In the UCS architecture, users served by different-size clusters are scheduled on distinct RBs.  For the example in Table~\ref{sample-RBs}, such an architecture enables scheduling policies with RBs of types \#1, \#2, and \#3, but not of  type \#4.  

The NUM subject to the UCS architecture is 
\begin{subequations}\label{NUM-first}
\begin{align}
\max\limits_{\lambda_{\Csize}, x_{k\Cc}} \    & \sum_{k\in\mathcal{U}} \ U\!\left( \sum_{\Cc:\ |\Cc|\le \Csize_{\rm max}} x_{k\Cc} \, r_{k\Cc}  \right) \label{obj-NUM-first}\\
\text{s.t. } & \!\!\! \sum_{
\substack{
\Cc: \, j\in \Cc\\  |\Cc|=\Csize}} \, \sum_{k\in\mathcal{U}} \! x_{k\Cc}  \leq  \lambda_\Csize S_j(\Csize), \, \forall j, \, \Csize \!\le\! \Csize_{\max},  \label{ct-cluster-NUM-first}\\
& \sum_{\Cc:\ |\Cc| = \Csize}\!\!\! \!\! x_{k\Cc} \leq \lambda_\Csize, \ \forall k\in\mathcal{U}, \  \Csize \le \Csize_{\max}, \label{ct-ue-NUM-first}\\
& x_{k\Cc} \geq 0,\ \  \forall k\in\mathcal{U},\, \forall \Cc \, \text{with}\, |\Cc| \le \Csize_{\max}, \label{ct-x-NUM-first}\\
&\sum_{\Csize=1}^{\Csize_{\max}} \lambda_\Csize \leq 1, \label{ct-sumArch-NUM-first}\\
& \lambda_\Csize \geq 0, \ \ \  \forall \Csize \le \Csize_{\max}.\label{ct-Arch-NUM-first}
\end{align}
\end{subequations}
Ineq.~(\ref{ct-cluster-NUM-first}) signifies that the total activity fractions  of users served by BS $j$  in clusters of size $\Csize$ cannot exceed the product of available RBs and the maximum number of beams that can be spatially multiplexed at BS $j$ in clusters of size $\Csize$. 
Ineq.~(\ref{ct-ue-NUM-first}) signifies that the fraction of RBs over which user $k$ is served by clusters of size $\Csize$ cannot exceed the RB fraction allocated to size-$\Csize$ clusters.

It is easy to verify  that (\ref{NUM-first}) is a convex optimization problem.  Also note that, for $\Csize_{\max} =1$,  (\ref{NUM-first}) specializes to  the cellular massive MIMO NUM problem studied in \cite{BetBur14a}.

The second architecture we consider also allows serving users  of the type of RB \# 4. 
\begin{defi}\label{second-arch}
{\bf Mixed Cluster-Size Architecture (MCS)}: 
A scheme from Defn.~\ref{admissible-dmimo} is a MCS architecture, if 
 a $\lambda_\Csize\ge 0 $ fraction of the RBs is allocated  $\forall \Csize\in \{2,\,3,\,\ldots,\, \Csize_{\rm max}\}$, and if within any RB that is part of the  $\lambda_\Csize$ fraction the following are satisfied: (i) each scheduled user is served either in cellular mode, or by a (user-dependent) cluster of $\Csize$ BSs; (ii)  $\forall j\in \Jc$,  BS $j$ serves either at most $S_j$ users all in cellular mode, or at most $S_j(\Csize)$ users, all served in clusters of size $\Csize$.
\end{defi}
A convex NUM problem  analogous to (\ref{NUM-first}) can be formulated for the MCS architecture.

\section{Scheduling Policies for NUM Solution} \label{sec:scheduling-policies}
In this section, we investigate scheduling policies that yield  $\{x_{k\Cc}\}$ closely matching the solution of (\ref{NUM-first}).

\begin{defi}\label{feas-sched-one} {\bf Feasible Schedule:}
A scheduling policy 
$\{ \Sc_\Cc(t); \ \  \forall \Cc, \text{ with } |\Cc|\le \Csize_{\rm max},   \forall t \in \{1, \, 2\, \cdots, T\} \}$ is feasible with respect to the UCS architecture of Defn.~\ref{first-arch} if it satisfies the following: 
\begin{enumerate}[(i)]
\item For each $t$, the policy associates with RB $t$ a single cluster size, $\Csize(t)$, for some $\Csize(t)\le \Csize_{\max}$, i.e., for each $\Cc$ for which  $\Sc_\Cc(t)$ is non-empty,  $|\Cc|=\Csize(t)$.
\item For each $t$, each user is served by at most one cluster; that is, $|\sum_\Cc\,1\{ k\in  \Sc_\Cc(t)\}|\le 1$ for all $k\in \Uc$.
\item For each $t$,  and for each  $j\in \Jc$, BS  $j$ serves at most $S_{j}(\Csize(t))$ users; that is,
$| \cup_{\Cc:\, j\in \Cc} \Sc_\Cc(t)| \le S_{j}(\Csize(t))$.
\end{enumerate}
\end{defi}
It is easy to verify that any feasible schedule yields activity fractions that satisfy  (\ref{ct-cluster-NUM-first})--(\ref{ct-Arch-NUM-first}).   



For instance, in a network of $3$ BSs,  with $\Csize_{\rm max}=2$ and $S_j(2)=3, \forall j$, no feasible schedule yields $\{x_{k\Cc}\}$  with $\lambda_2>0$, for which (\ref{ct-cluster-NUM-first}) is satisfied with equality for all $j$ and $\Csize=2$.  This is because it is impossible to {\em simultaneously} schedule $3$ users  at {\em all} three BSs: at most two  BSs can schedule 3 users, while the 3rd would necessarily schedule at most 2 users (i.e., the 3 BSs schedule a total of 4 users, each receiving beams from a BS pair).  Clearly, any feasible schedule results in at least one {\em strict} inequality in (\ref{ct-cluster-NUM-first}). 
Hence, the coarser time-scale NUM problem (\ref{NUM-first}) does not capture the finer time-scale  constraints associated with feasible schedulers. Although, in general, (\ref{NUM-first}) provides an upper bound on the network performance, as we show next,   using activity fractions that are the solution to (\ref{NUM-first}), we can design scheduling policies, whose performance is close to the utility provided by the solution to~(\ref{NUM-first}).


\subsection{Virtual Queue Based Scheduling Scheme}\label{sec:vq-scheme}
 As in \cite{YeAndLB13,BetBur14a}, we focus on the proportional fair utility (i.e.,   $U(x)=\log(x)$ in (\ref{obj-NUM-first})) in the rest of this paper.
We consider scheduling policies for the UCS architecture  comprised of $\Csize_{\max}$ parallel schedulers, one per  cluster size $\Csize\in \{1,\, 2,\, \cdots, \Csize_{\rm max}\}$.  
We describe a method for scheduling users over the RBs from the $\lambda_\Csize>0$ fraction of RBs dedicated to clusters of size $\Csize$.

We first remark, that as in the cellular settings \cite{YeAndLB13,BetBur14a}, empirical evidence reveals that, in a ``loaded'' network, most users  are uniquely associated to a single cluster per cluster size, i.e.,  for most user indices $k$, there is a single nonzero $x_{k\Cc}$  among all  $\Cc$'s  with the same  $|\Cc|$.

Insight regarding this observation can be obtained by examining \emph{Karush-Kuhn-Tucker} (KKT) conditions of (\ref{NUM-first}), which imply
\begin{equation}\label{eq:kkt-1}
\sum_{\Csize'}  \sum_{\Cc: \, |\Cc| = \Csize'} x_{k\Cc} r_{k\Cc} \geq \frac{r_{k\Cc}}{ \sum_{j} \nu_{j\Csize} +  \mu_{k\Csize}},
\end{equation}
where $\nu_{j\Csize}$ and $\mu_{k\Csize}$ are the Lagrange multipliers corresponding to (\ref{ct-cluster-NUM-first}) and (\ref{ct-ue-NUM-first}), respectively.

In a  loaded network, where the constraints (\ref{ct-ue-NUM-first}) are inactive (i.e., $\sum_{\Cc:\, |\Cc|= \Csize} x_{k\Cc}< \lambda_\Csize$ $\forall k\in \mathcal{U}$), we have the following:
\begin{prop}\label{prop:uniqueass}
If (\ref{ct-ue-NUM-first}) are inactive $\forall k\in\mathcal{U}$, the number of users that are served by multiple clusters of size $\Csize$ is at most $N_\Csize -1$, where $N_\Csize$ is the number of size-$\Csize$ clusters.
\end{prop}
\begin{proof}
See Appendix~\ref{pf:prop-uniqueass}.
\end{proof}

Given the limited number of fractional users per cluster size $\Csize$, the scheduler approximates the optimal  $\{x_{k\Cc}\}$ by unique association activity fractions,   $\{\tilde{x}_{k\Cc}\}$, given by
\begin{equation}
\label{xtil-kC}
 \tilde{x}_{k\Cc}=
\begin{cases} 
 x_{k\Cc} & \text{if  $\Cc = \Cc^*(k) $}\\
 0 &  \text{otherwise} \end{cases},
\end{equation}
with $
\Cc^*(k) = \argmax_{\Cc: \ |\Cc|=L}  x_{k\Cc}$.

Letting $\mathcal{U}_\Cc$ denote the users for which $\tilde{x}_{k\Cc}>0$, we  have  $\mathcal{U}_\Cc\cap\mathcal{U}_{\Cc'}=\emptyset$  for all   $\Cc\neq \Cc'$ with $|\Cc|=|\Cc'|$. We also let $\mathcal{U}^{(\Csize)} = \cup_{\Cc:\, |\Cc|=\Csize} \, \mathcal{U}_{\Cc}$ denote the set of users that receive non-zero activity from clusters of size $\Csize$.  

To assign user $k$ a fraction of RBs close to the desired fraction $\alpha_k =\tilde{x}_{k\Cc}/\lambda_\Csize$, we consider a max-min scheduling  policy based on virtual queues (VQ), which assumes user $k$ receives rate $\tilde{R}_{k} = 1/\alpha_k$ when user $k$ is scheduled for transmission over cluster $\Cc^*(k)$ (i.e., $k\in \Sc_{\Cc^*(k)}(t)$). 
The  cluster-size $\Csize$  scheduler performs at each $t$ a weighted sum rate maximization (WSRM) of the form \cite{ShiCai10}:
\begin{subequations}\label{VQ-C}
\begin{align}
\max_{\tilde{\Uc} \subseteq \Uc^{(\Csize)}}\ & \ \sum_{ k \in \tilde{\Uc}  } Q_k(t)\tilde{R}_{k},
\label{VQ-C-WSRM}
\\
\text{s.t. } \ \ \  &  \  \   \sum_{k\in \tilde{\Uc}} 1\{j\in \Cc^*(k)\} \le S_{j}(\Csize), \  \ \forall j\in \Jc,
\label{VQ-C-constraints}
\end{align}
where the weight of user $k$ at time $t$, $Q_k(t)$, is  the VQ length of user $k$ at time $t$. For max-min fairness \cite{ShiCai10},   $Q_k(t)$ is updated as follows: 
\begin{equation}
Q_k(t+1) = \max\{0, Q_k(t) - \tilde{R}_k(t)\} + A_k(t),
\end{equation}
where
\begin{equation}
\tilde{R}_k(t) = \begin{cases} 
 \tilde{R}_k  & \text{if user $k$ is scheduled at time $t$}\\
 0 &  \text{otherwise} \end{cases},
 \end{equation}
\begin{equation}
A_k(t)  = \begin{cases} 
 A_{\rm max}  & \text{if $V > \sum_k Q_k(t)$} \\
 0 &  \text{otherwise} \end{cases},
 \end{equation}
\end{subequations}
with $A_{\max}$ and $V$ chosen sufficiently large \cite{ShiCai10}. 
Note that in the absence of  constraints (\ref{VQ-C-constraints}), the max-min scheduler  (\ref{VQ-C})  schedules user $k$ the desired fraction of RBs, $\alpha_k$.


Scheduling via (\ref{VQ-C}) is impractical, as it amounts to solving for each RB $t$ an integer linear program of the form (\ref{VQ-C-WSRM})--(\ref{VQ-C-constraints}). A number of heuristic algorithms can be used to provide feasible (though generally suboptimal) solutions to (\ref{VQ-C}).  In this paper, we consider a rudimentary greedy algorithm. Letting $K_\Csize= |\Uc^{(\Csize)}|$ be the total number of users to be served by clusters of size $\Csize$, the greedy algorithm for size-$L$ clusters at time $t$ operates as follows:
\begin{enumerate}[1.]
\item Determine a user order $\pi(k)$, where  $Q_{\pi(k)}(t) \tilde{R}_{\pi(k)} \ge Q_{\pi(k+1)}(t) \tilde{R}_{\pi(k+1)}$  for all $k\in\Uc^{(\Csize)}$.
\item Initialization: $k=1$, and $\tilde{\Uc}=\emptyset$.
\item  If the user set $\tilde{\Uc} \cup \{\pi(k)\}$ satisfies all the constraints in (\ref{VQ-C-constraints}), set $\tilde{\Uc} =\tilde{\Uc} \cup \{\pi(k)\}$.
\item If $k< K_\Csize$,  set $k=k+1$ and go to step 3.
\item Output $\tilde{\Uc}$ as the scheduling user set  at time $t$.
\end{enumerate}

\section{Performance Evaluation}\label{sec:simulation}
In this section, we present a brief simulation-based evaluation of the proposed distributed MIMO schemes based on the ``wrap-around'' checkerboard layout in Fig.~\ref{fig:network}. There are 4 macros with $M_j=100$ and $S_j(L)=10L$,  and 32 pico BSs  with  $M_j=40$ and $S_j(L)=4L$.  One pico BS is at the center of each white square, while 3 pico BSs are dropped uniformly within each shaded square.  Also, 15 and 90  single-antenna users are dropped uniformly in each white and each shaded square, respectively.  The macro and pico BS transmit powers are 46dBm and 35dBm, respectively. The path-loss for  macro-user links and pico-user links are $128.1+37.6\log_{10}d$ and $140.7+36.7\log_{10}d$,  respectively, with the distance $d$ in km.  

We consider two distinct macro-pico operation scenarios: (i) macros and picos operate on the same band, with  cluster sizes up to $\Csize_{\rm max}=4$; (ii) macros are given  20\% of the RBs  for cellular transmission, and picos are given  the remaining 80\%  for distributed MIMO with $\Csize_{\rm max}=4$.

%
%

\begin{figure}
\centering
\includegraphics[width=6.8cm, height=6.8cm]{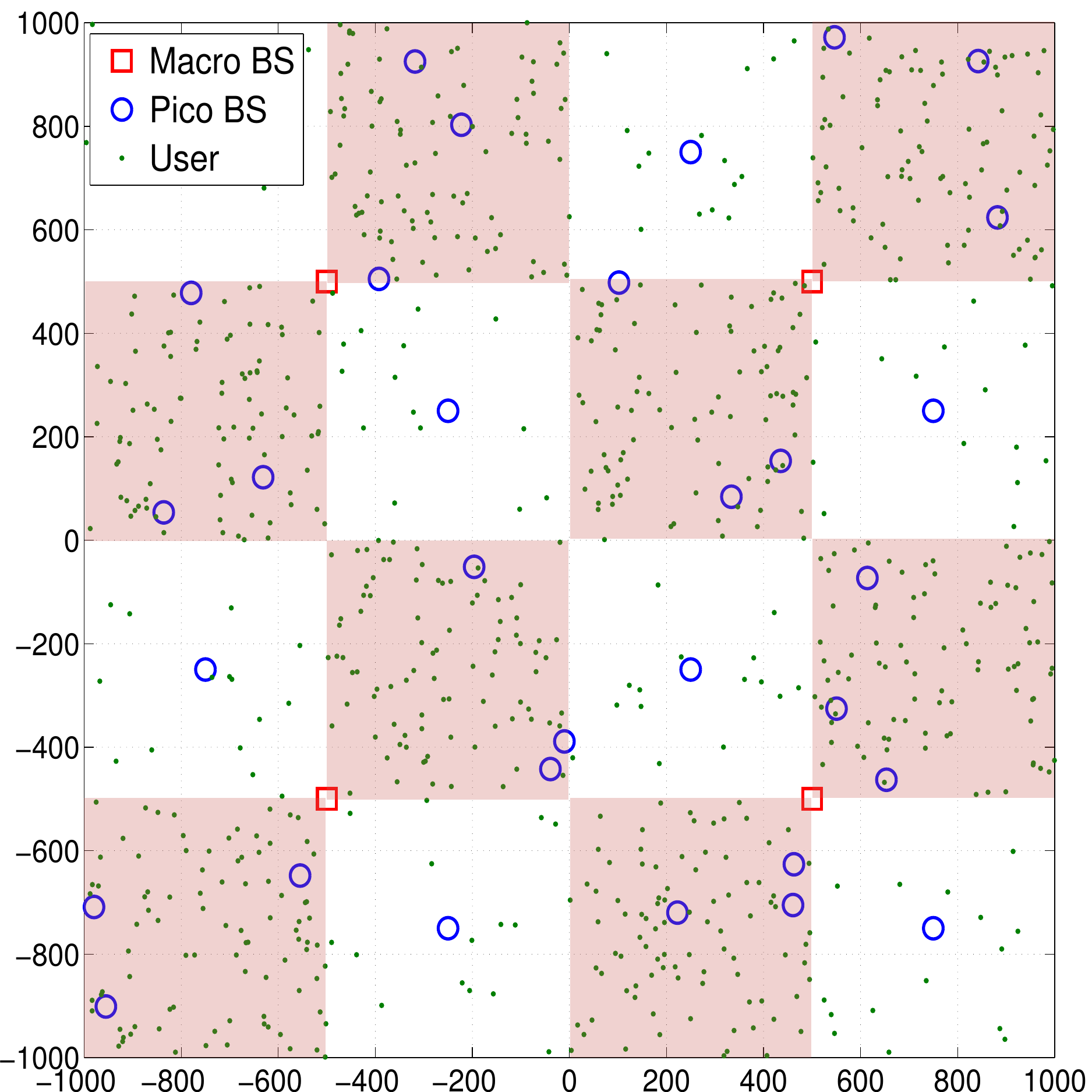}
\caption{A 2000m $\times$ 2000m network with 4 macro 32 pico BSs.}
\label{fig:network}
\end{figure}

Figs.~\ref{fig:geomean-rate} and \ref{fig:Rcdf} compare the proposed distributed MIMO schemes\footnote{The  NUM problem for scenario (ii) is a simple extension of (\ref{NUM-first}).} against network-optimized cellular transmission \cite{BetBur14a} and max-SINR based association.
Fig. \ref{fig:geomean-rate} shows the user-rate geometric mean for each scheme and  each operation scenario considered. 
As the figure reveals,  unique-association  (i.e., the $\{\tilde{x}_{k\Cc}\}$) yield almost the same performance as the optimal solution, verifying our conclusion that the number of users served by multiple clusters per architecture is limited.  Also, the proposed greedy VQ based scheduler performs within  90\%  of the NUM optimal value.  More importantly, it significantly outperforms network-optimized cellular operation under both scenarios. Note that though the orthogonal resource allocation with optimal user association in cellular case performs better than the shared resource allocation in our setting, which operation scenario is better highly depends on the system parameters (e.g., channel mode, transmit power and BS density).

\begin{figure}
\centering
\includegraphics[width=7.6cm,height=6.8cm]{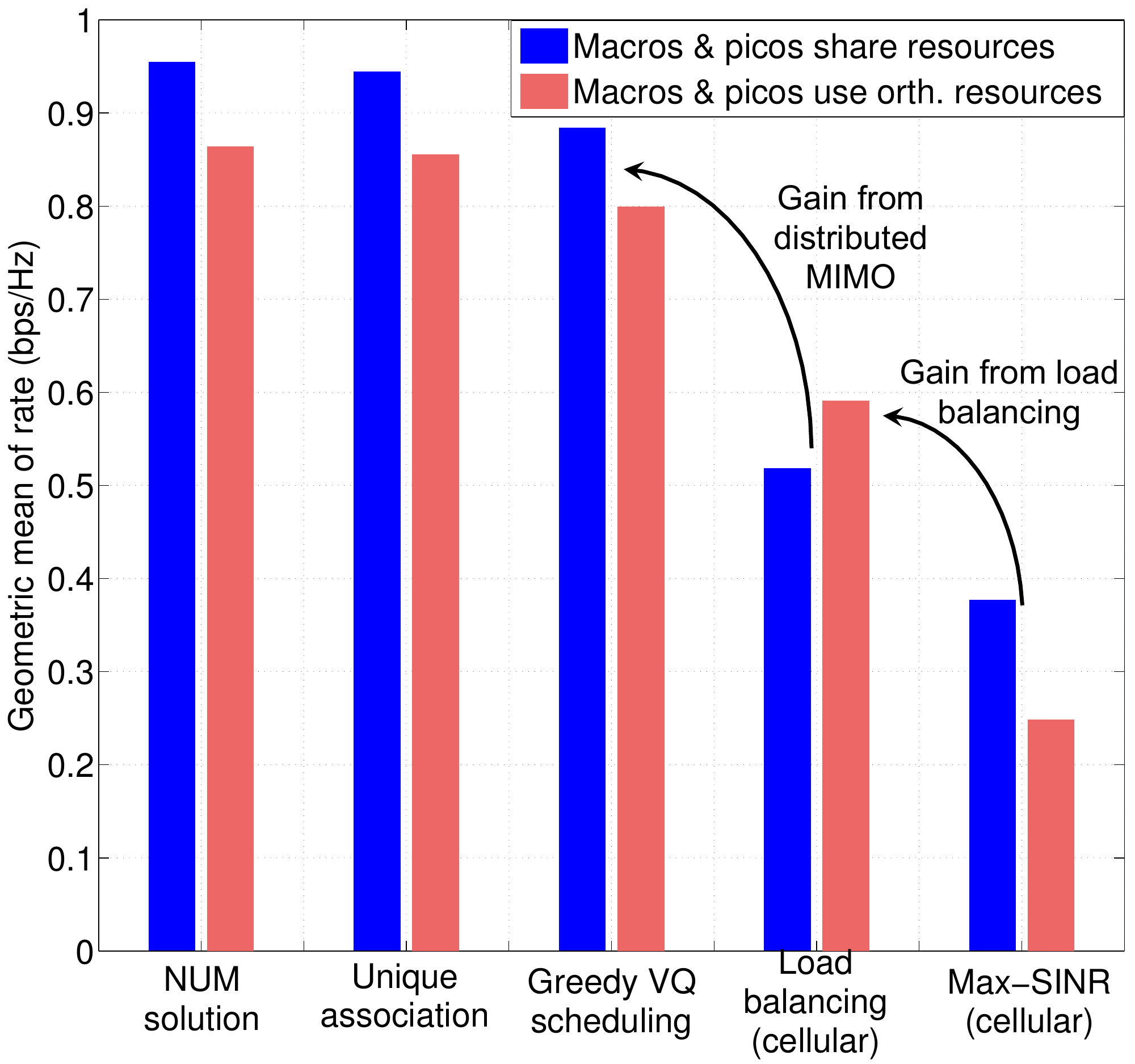}
\caption{User-rate geometric means for various schemes under two distinct macro-pico operational scenarios.}
\label{fig:geomean-rate}
\end{figure}

Fig. \ref{fig:Rcdf} shows the corresponding user-rate cumulative distribution functions.  As the figure shows,  the proposed distributed MIMO schemes yield about a  $2\times$ gain in 5th percentile rates with respect to the optimal cellular scheme \cite{BetBur14a}, under both macro-pico operation scenarios.
 \begin{figure}
 \centering
		\subfigure[Macro and pico BSs share resources]{
			\label{fig:Rcdf-share}
			\includegraphics[width=7.6cm, height=6.8cm]{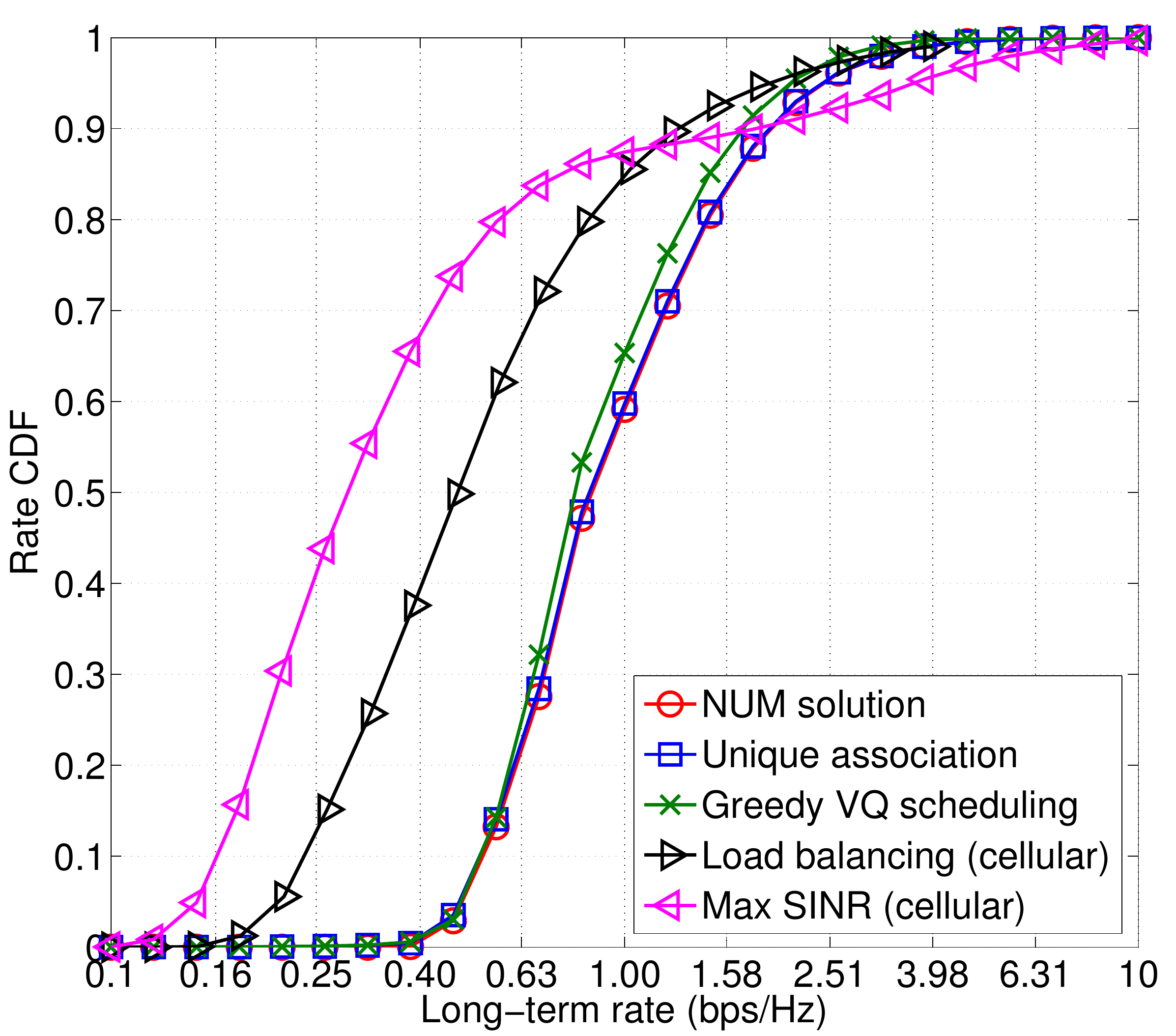}}
		\subfigure[Macro and pico BSs use disjoint sets of resources]{
			\label{fig:Rcdf-orth}
			\includegraphics[width=7.6cm, height=6.8cm]{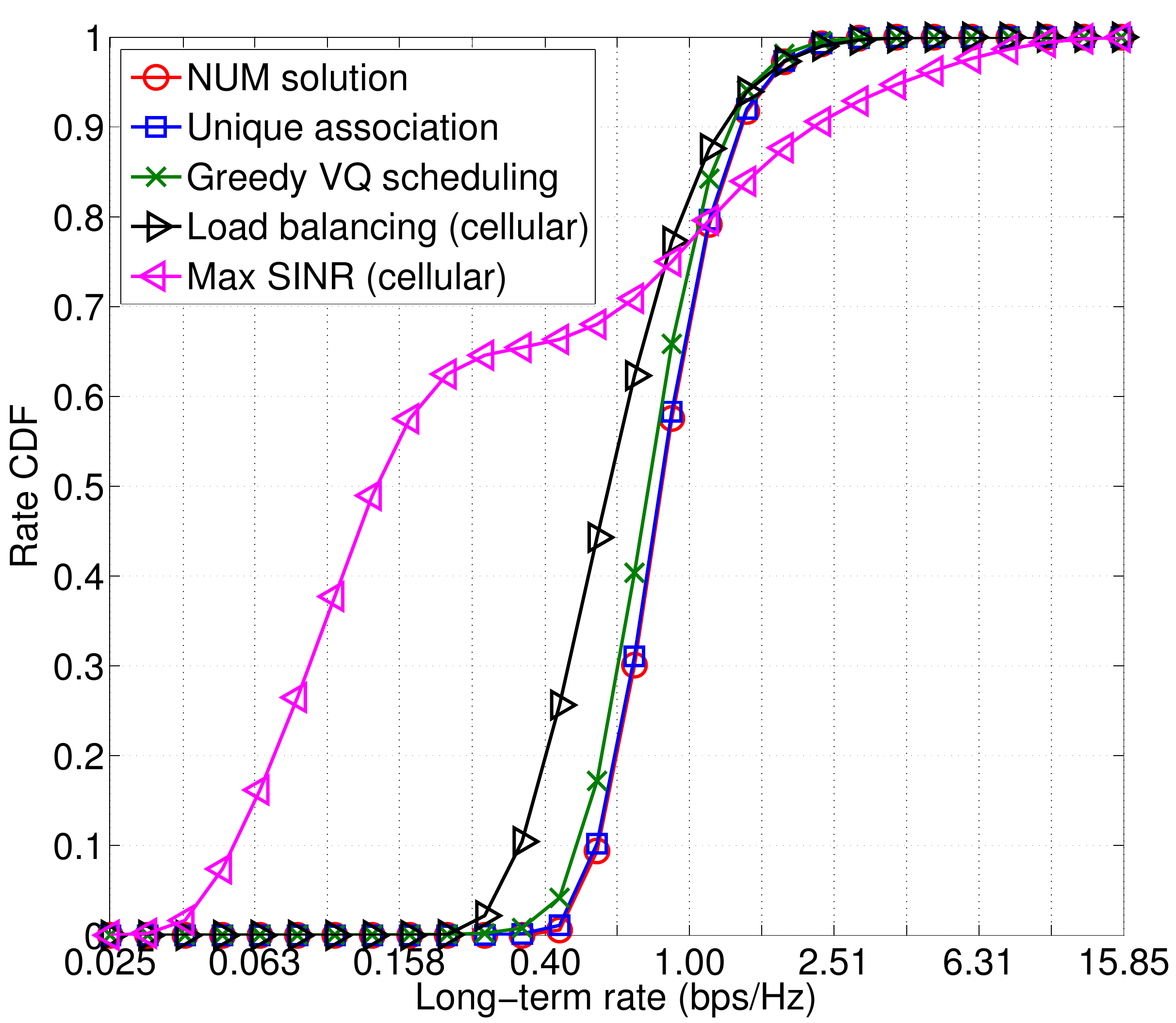}}
 \caption{User-rate CDFs for various schemes.}
\label{fig:Rcdf}
\end{figure}

\section{Conclusion}\label{sec:conclusion}
We present techniques for harmonized use of cellular and CoMP transmission over massive MIMO HetNets. The techniques rely on using a class of distributed MIMO transmission schemes, which do not require CSI exchanges among BSs, and can enable flexible (user-specific) CoMP transmission.   We use properties of the distributed MIMO schemes in the massive MIMO regime to formulate resource allocation as a convex NUM problem, and present scheduling policies whose goal is to approximate the resulting optimized resource allocations. 
As our simulations show,  the proposed operation offers significantly performance gains  with respect to the  network-optimized cellular-only massive MIMO operation \cite{BetBur14a}, especially at the cell edge. More dynamic settings (e.g., users with high mobility) are left for future work. The investigation of other simulation settings (e.g., different $\rho_C$) is also of interest.


\bibliographystyle{ieeetr}
\bibliography{massivemimobib}


\appendices

\section{Proof of Spectral Efficiency Using ZF Precoding}\label{pf:theo-zf}
In this paper, we assume that each BS has perfect CSI.  Techniques in \cite{HuhTul12,LimCha13} can be applied to derive the results. For completeness, we provide the proof as follows. 
We use  $\mathbb{E}\left[\frac{S}{I+N}\right] \approx \frac{\mathbb{E}\left[S\right]}{\mathbb{E}\left[I+N\right]}$  to approximate SINR in the calculation of ergodic spectral efficiency in the massive MIMO regime, which is shown to be quite close to the exact asymptotic spectral efficiency \cite{LimCha13}.

Adopting the ZF precoding, the precoding matrix at BS $j$ is $\mb{F}_j=\mb{G}_j\left(\mb{G}_j^H\mb{G}_j\right)^{-1}\mb{A}_j^{1/2}$, where $\mb{A}_j$ is the normalizing coefficients matrix. In this case, the intra-cell interference is $0$.  Denoting the $k$th diagonal element of $\mb{A}_j$ by $a_{kj}$ and plugging the precoding matrix $\mb{F}_j=\mb{G}_j\left(\mb{G}_j^H\mb{G}_j\right)^{-1}\mb{A}_j^{1/2}$ into received signal, the SINR at user $k$ from $\Cc$ is
\begin{equation}\label{eq:sinrexact}
\begin{aligned}
&\sinr_{k\Cc}\\
=  &\frac{\sum_{j\in \Cc} \sum_{l\in \Cc} \sqrt{\frac{P_{j}P_{l}a_{kj}a_{kl}}{ S_{j}(|\Cc|)  S_{l}(|\Cc|)} }}{\sigma^2 + \left\| \sum_{l\notin \Cc} \sum_{u\in \cup_{(\mathcal{C'}: l\in \mathcal{C'})} \Sc_\mathcal{C'}}  \sqrt{\frac{P_l}{ S_{l}(\mathcal{C'})}} \mb{g}_{kl}^H \mb{f}_{ul}s_u\right\|^2},
\end{aligned}
\end{equation}
where $\mathcal{C'}$ denotes the cluster including $l$ that is different from $\Cc$.
Using similar techniques in the proof of Theorem III-1 in~\cite{WonPan08}, we can show  $\frac{a_{kj}}{S_{j}(|\Cc|)}\rightarrow \beta_{kj}\frac{\left(M_j - S_{j}(|\Cc|) +1\right)}{S_{j}(|\Cc|)}$, as $M_j\rightarrow \infty$ with fixed ratio $S_{j}(|\Cc|)/M_j\leq 1$. Then we have
$\sum_{j\in \Cc} \sum_{l\in \Cc} \sqrt{\frac{P_{j}P_{l}a_{kj}a_{kl}}{S_{j}(|\Cc|)S_{l}(|\Cc|)} } \rightarrow \sum_{j\in \Cc} \sum_{l\in \Cc} \sqrt{P_{j}P_{l}\beta_{kj}\beta_{kl}b_jb_l }$, where $b_j=\frac{M_{j} - S_{j}(|\Cc|) +1}{S_{j}(|\Cc|)}$. As for the interference, we have
\begin{equation*}
\begin{aligned}
&\left\| \sum_{l\notin \Cc} \sum_{u\in \cup_{(\mathcal{C'}: l\in \mathcal{C'})} \Sc_\mathcal{C'}}   \sqrt{\frac{P_l}{S_{l}(|\Cc|)}} \mb{g}_{kl}^H \mb{f}_{ul}s_u\right\|^2 \\
=&\sum_{l\notin \Cc} \sum_{u\in \cup_{(\mathcal{C'}: l\in \mathcal{C'})} \Sc_\mathcal{C'}}  \left\|   \sqrt{\frac{P_l}{S_{l}(|\Cc|)}} \mb{g}_{kl}^H \mb{f}_{ul}\right\|^2\\
\rightarrow &\sum_{l\notin \Cc} P_l \beta_{kl},
\end{aligned}
\end{equation*}
where the last step follows from that channels and precoders of different users are independent. Based on the approximation $\mathbb{E}\left[\frac{S}{I+N}\right] \approx \frac{\mathbb{E}\left[S\right]}{\mathbb{E}\left[I+N\right]}$, we complete the proof by plugging the above results into~(\ref{eq:sinrexact}).


\section{Proof of  Spectral Efficiency Using MRT Precoding}\label{pf:theo-mrt}
We first give the following properties of MRT in the massive MIMO regime.  

\noindent 1) We have $\|\mb{g}_{kj}\|^2=\mb{g}_{kj}^H\mb{g}_{kj} = \beta_{kj} \sum_{i=1}^{M_j} h_{kj,i}^*h_{kj,i}$. Recalling that $h_{kj,i}$ are i.i.d. Gaussian,  we have $\frac{1}{M_j}\|\mb{g}_{kj}\|^2 \rightarrow \beta_{kj}\mathbb{E}[h_{kj,1}^*h_{kj,1}]=\beta_{kj}$, as $M_j$ and $S_{j}(|\Cc|)$ become large with a fixed ratio $S_{j}(|\Cc|) /M_j\leq 1$.

\noindent 2) Plugging $\mb{f}_{kj}$, we have $\left|\mb{g}_{kj}^H\mb{f}_{jn}\right|^2=  \left|\mb{g}_{kj}^H \frac{\mb{g}_{nj}}{\|\mb{g}_{nj}\|}\right|^2=\left|\frac{\sqrt{\beta_{kj}\beta_{nj}}}{\|\mb{g}_{nj}\|}\sum_{i=1}^{M_j}h_{kj,i}^*h_{nj,i} \right|^2$, which converges to $\frac{\beta_{kj}\beta_{nj}}{\frac{1}{M_j}\|\mb{g}_{nj}\|^2}\mathbb{E}\left[|h_{kj,1}^*h_{nj,1}|^2\right] +M_j(M_j-1)\mathbb{E}\left[h_{kj,1}^*h_{jn,1}h_{kj,2}^*h_{nj,2} \right]=\beta_{kj}$ as $M_j\rightarrow\infty$, since $h_{kj,i}$ and $h_{nj,i}$ are i.i.d. Gaussian for $n\neq k$.

Using the above two properties and similar techniques in Appendix \ref{pf:theo-zf}, we have
\begin{equation}\label{eq:sinr-mrt}
\begin{aligned}
&\sinr_{k\Cc} \\
\approx & \frac{\left(\sum_{j\in\Cc} \sqrt{\frac{P_{j}}{S_{j}(|\Cc|) }} \|\mb{g}_{kj}\|  \right)^2 }{1 + \sum_{j\in\Cc} (S_{j}(|\Cc|) -1)\frac{P_j}{S_{j}(|\Cc|) }\beta_{kj} + \sum_{l\notin \Cc} P_l\beta_{kl}}\\
= & \frac{\sum_{j\in\Cc}\sum_{l\in\Cc} \sqrt{\frac{P_jP_lM_jM_l\beta_{kj}\beta_{kl}}{S_{jC}S_{lC}}} }{1 + \sum_{j\in\Cc} (S_{j}(|\Cc|) -1)\frac{P_j}{S_{j}(|\Cc|)}\beta_{kj} + \sum_{l\notin \Cc} P_l\beta_{kl}}.
\end{aligned}
\end{equation}
Plugging (\ref{eq:sinr-mrt}) into $\log_2(1+\sinr_k)$, we complete the proof. 

\section{Proof of Proposition \ref{prop:uniqueass}}\label{pf:prop-uniqueass}
We use the techniques similar to the proof of Proposition~3 in \cite{YeAndABS13}, where a graph is used to represent the association, and KKT conditions (\ref{eq:kkt-1})  restrict the structure of the graph.

We denote the graph by $G_1$, whose nodes represent the users, and the edge between two nodes represents the BS cluster that serves the two users in the considered architecture. Each node has an ID indicating the user index, while each edge has a color that identifies the BS cluster.  For example, Fig. \ref{fig:associate3} shows that user $k$ is served by both clusters $\Cc_1$ and $\Cc_2$, and user $m$ is served by both clusters $\Cc_1$ and $\Cc_3$. 

\begin{figure}
\centering
\includegraphics[width=4cm, height=4cm]{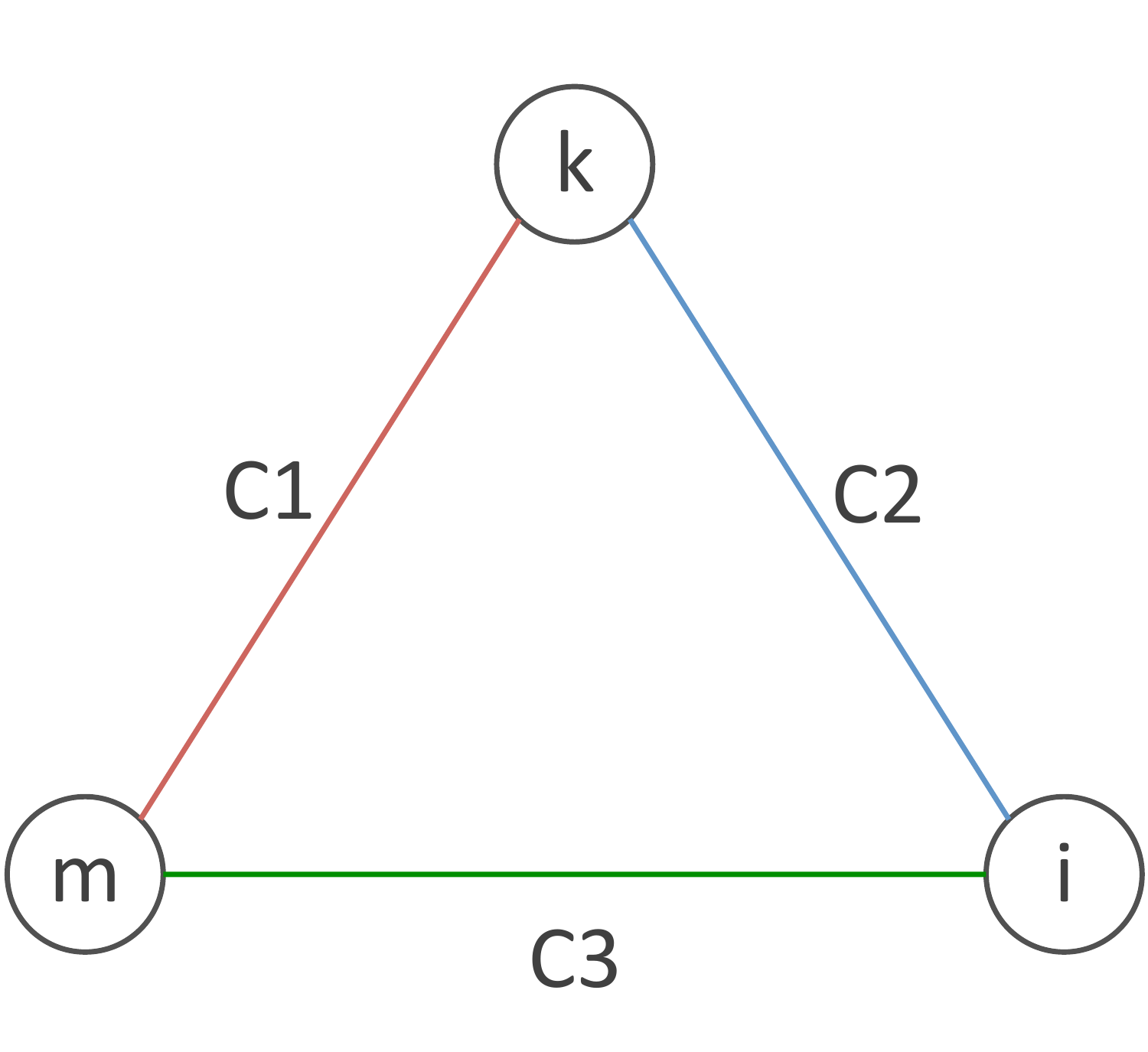}
\caption{The graph representation of the associations of three users.}
\label{fig:associate3}
\end{figure}

In a heavily loaded network, where the constraints (\ref{ct-ue-NUM-first}) are inactive (i.e., $\sum_{\Cc:\, |\Cc|= \Csize} x_{k\Cc}< \lambda_\Csize$) in the optimal solutions, we have $\mu_{k\Csize}=0$ for all $k$. 
If there are two users $k$ and $m$ being served by size-$\Csize$ clusters $\Cc_1$ and $\Cc_2$  (i.e., $x_{k\Cc_1}>0$, $x_{k\Cc_2}>0$, $x_{m\Cc_1}>0$, $x_{m\Cc_2}>0$),  we have  $R_k= \frac{r_{k\Cc_1}}{ \sum_{j\in \Cc_1}\nu_{j\Csize}} =\frac{r_{k\Cc_2}}{ \sum_{j\in \Cc_2}\nu_{j\Csize}} $ and $R_m= \frac{r_{m\Cc_1}}{ \sum_{j\in C_1}\nu_{j\Csize}} =\frac{r_{m\Cc_2}}{ \sum_{j\in \Cc_2}\nu_{j\Csize}} $ from KKT condition (\ref{eq:kkt-1}), where $R_k= \sum_{\Cc} x_{k\Cc} r_{k\Cc}$. Thus,  we have
\begin{equation}\label{eq:kkt-2user}
\frac{r_{k\Cc_1}}{r_{k\Cc_2}}=\frac{r_{m\Cc_1}}{r_{m\Cc_2}},
\end{equation}
which is true with probability  $0$.  Therefore, it is almost sure that any two users can share at most one same cluster in each architecture. Similarly, we consider an example of three users $k, m, i$ and clusters $\Cc_1, \Cc_2, \Cc_3$ as illustrated in Fig. \ref{fig:associate3}. We consider the following three cases:

\noindent 1) If clusters $\Cc_1$, $\Cc_2$ and $\Cc_3$ are different, we have
\begin{equation}
\begin{aligned}
\frac{r_{k\Cc_1}}{r_{k\Cc_2}}&=\frac{\sum_{j\in \Cc_1}\nu_{j\Csize}}{\sum_{j\in \Cc_3}\nu_{j\Csize}} \frac{\sum_{j\in \Cc_3}\nu_{j\Csize}}{\sum_{j\in \Cc_2}\nu_{j\Csize}} \\
&=\frac{r_{m\Cc_1}}{r_{m\Cc_3}}\frac{r_{i\Cc_3}}{r_{i\Cc_2}},
\end{aligned}
\end{equation}
which is true with probability $0$.

\noindent 2) If $\Cc_1=\Cc_2\neq \Cc_3$, we have that users $m$ and $i$ are served both by clusters $\Cc_1$ and $\Cc_3$, which is true with probability $0$ from (\ref{eq:kkt-2user}).

\noindent 3) If $\Cc_1=\Cc_2=\Cc_3$, we have that users $k$, $m$ and  $i$ are served by the same cluster, which is possible. In this case, the graph  becomes a \emph{complete graph}.

Therefore, the graph $G_1$ with three users either contains a loop with the same color edges or no loop. We can get a similar result for graph $G_1$ with more than three users, where the users served by the same BS cluster constitute a complete graph. Thus, we generate a new graph, denoted by $G_2$, where the node represents a cluster. There is an edge between two nodes in $G_2$, if these two nodes (i.e., clusters) have a common vertex in $G_1$ (i.e., there is at least one user served  by both these two clusters). Thus the number of users who are served by more than one cluster is limited by the edge of $G_2$. Note that there are $N_\Csize$ nodes and no loop in $G_2$. Thus, $G_2$ is a tree, which has the  maximal number of edges being one less than the number of nodes (i.e., $N_\Csize-1$). Therefore, the number of users served by multiple BS clusters equals the number of edges in graph $G_2$, which is no more than $N_\Csize-1$.  

\end{document}